\documentclass[english]{cccconf}
\usepackage{epstopdf}

\usepackage{amsmath,amssymb,amsfonts}
\usepackage{amsthm}
\usepackage[ruled,linesnumbered]{algorithm2e}
\usepackage{array}
\usepackage{booktabs}
\usepackage{color}
\usepackage{cite}
\usepackage{cuted}
\usepackage{enumitem}
\usepackage{graphicx}
\usepackage{hyperref}
\usepackage{mathrsfs}
\usepackage{multirow}
\usepackage[caption=false,font=scriptsize,labelfont=sf,textfont=sf]{subfig}
\usepackage{textcomp}
\usepackage{threeparttable}
\usepackage{stfloats}
\usepackage{upgreek}
\usepackage{url}
\usepackage{verbatim}

\newcommand\dif{\mathrm{d}}

\newtheorem{theorem}{Theorem}

\begin{document}

\title{Analyzing the Crowding-Out Effect of Investment Herding on Consumption: An Optimal Control Theory Approach}

\author{Huisheng Wang and H. Vicky Zhao}

\affiliation{Department of Automation,
        Tsinghua University, Beijing 100084, P.~R.~China
        \email{\href{whs22@mail.tsinghua.edu.cn}{whs22@mails.tsinghua.edu.cn}, \href{vzhao@tsinghua.edu.cn}{vzhao@tsinghua.edu.cn}}}
        
\maketitle

\begin{abstract}
Investment herding, a phenomenon where households mimic the decisions of others rather than relying on their own analysis, has significant effects on financial markets and household behavior. Excessive investment herding may reduce investments and lead to a depletion of household consumption, which is called the crowding-out effect. While existing research has qualitatively examined the impact of investment herding on consumption, quantitative studies in this area remain limited. In this work, we investigate the optimal investment and consumption decisions of households under the impact of investment herding. We formulate an optimization problem to model how investment herding influences household decisions over time. Based on the optimal control theory, we solve for the analytical solutions of optimal investment and consumption decisions. We theoretically analyze the impact of investment herding on household consumption decisions and demonstrate the existence of the crowding-out effect. We further explore how parameters, such as interest rate, excess return rate, and volatility, influence the crowding-out effect. Finally, we conduct a real data test to validate our theoretical analysis of the crowding-out effect. This study is crucial to understanding the impact of investment herding on household consumption and offering valuable insights for policymakers seeking to stimulate consumption and mitigate the negative effects of investment herding on economic growth.
\end{abstract}


\section{Introduction}
\label{sec:introduction}
In macroeconomics, investment and consumption are two of the most crucial decisions for households. \textit{Investment} refers to the allocation of funds among financial assets, aiming to maximize returns while minimizing risk, while \textit{consumption} involves using available funds for immediate needs and desires \cite{barro1997macroeconomics}. When making investment decisions, households often exhibit \textit{investment herding}, a phenomenon where households mimic the investment decisions of others rather than relying on their analysis \cite{scharfstein1990herd, chiang2010empirical}. Investment herding is well-documented in financial markets and can lead to significant negative consequences, such as bubbles \cite{chang2014herd}, market inefficiencies \cite{cipriani2014estimating}, and increased volatility \cite{huang2015herd}. Furthermore, excessive investment herding may reduce investments, leading to a depletion of household funds, which in turn diminishes household consumption \cite{bikhchandani2000herd}. This phenomenon, referred to as the \textit{crowding-out effect} of investment herding on consumption, is of particular concern in today's economic environment, where stimulating consumption is critical for economic growth \cite{nie2016consumer}. Policymakers can potentially enhance household consumption by addressing the crowding-out effect of investment herding. However, before making such policies, it is essential to first analyze the impact of investment herding on household consumption.

Many existing studies have qualitatively explored the impact of investment herding on household consumption in the field of macroeconomics. The work in \cite{komalasari2022herding} suggests that financial market instability, exacerbated by investment herding, can affect real economic activity, particularly during financial crises, thereby reducing household funds and consumption decisions. The work in \cite{yoon2022investor} demonstrates that social media sentiment is associated with investment herding, which can influence households' judgments and, consequently, their consumption decisions, particularly regarding information dissemination and emotional contagion. The authors of \cite{yang2024unpacking} delve into how investment herding is influenced by individual characteristics and personality traits like impulsivity and risk-taking, which are correlated with investment herding and may affect consumption decisions. Despite these valuable insights, to the best of our knowledge, few studies have quantitatively analyzed the crowding-out effect of investment herding on household consumption.

Based on the optimal control theory, the optimal investment and consumption theory offers a quantitative approach to studying how households dynamically adjust their investment and consumption decisions \cite{1952Portfolio, samuelson1975lifetime}. As a classical model in optimal investment and consumption theory, the Merton problem involves a household allocating funds to a risky asset to maximize the expected utility of terminal wealth, while simultaneously factoring in consumption decisions \cite{merton1969lifetime}. Extending the Merton problem, our prior work in \cite{wang2024ccc} quantitatively analyzes the impact of investment herding on household investment decisions, and in this work, we call it the investment-only problem. However, few studies have jointly considered both investment herding and consumption decisions, nor have they analyzed the crowding-out effect within a unified framework.

To address these issues, in this work, we first formulate the optimal investment and consumption problem under the impact of investment herding and solve for the optimal investment and consumption decisions of households based on the optimal control theory. Then, we theoretically analyze the impact of investment herding on household consumption and demonstrate the existence of the crowding-out effect. 

The remainder of this paper is organized as follows. In Section \ref{sec:problem-definition}, we formulate the optimal investment and consumption problem. In Section \ref{sec:solution}, we derive the analytical solutions for the optimal investment and consumption decisions of households. We quantitatively analyze the impact of investment herding on consumption decisions and demonstrate the existence of the crowding-out effect in Section \ref{sec:theoretical-analysis}. In Section \ref{sec:real-data-validation}, we conduct a real data test to validate our theoretical analysis. Finally, Section \ref{sec:conclusion} concludes the paper.

\section{Problem Definition}
\label{sec:problem-definition}
In this work, we assume that the financial market consists of one risk-free asset with an interest rate $r$ and one risky asset with an excess return rate $v$ and volatility $\sigma$. We consider two households, $\mathsf{H}_1$ and $\mathsf{H}_2$, where $\mathsf{H}_1$ is the following household and $\mathsf{H}_2$ is the leading household whose investment decisions unidirectionally influence those of $\mathsf{H}_1$ due to investment herding. The important notations in this work are summarized in Table \ref{tab:notation}. We denote the fund process of $\mathsf{H}_i, i \in \{1, 2\}$ as $\{X_i(t)\}_{t \in \mathcal{T}}$, where $X_i(0) = x_i$ represents the initial fund of household $\mathsf{H}_i$, and $\mathcal{T} := [0, T]$ is the investment and consumption time horizon. 

Different from the investment-only problem, at each time $t$, both households allocate their funds across three categories: 1) consumption, 2) investing in the risky asset, and 3) investing in the risk-free asset. The \textit{consumption decision} of $\mathsf{H}_i$, denoted by $\{C_i(t)\}_{t \in \mathcal{T}}$, represents the fund used to consume. We assume that $\mathsf{H}_i$'s consumption decision $\{C_i(t)\}_{t \in \mathcal{T}}$ is always positive. The \textit{investment decision} of $\mathsf{H}_i$, denoted by $\{I_i(t)\}_{t \in \mathcal{T}}$, refers to the fund allocated to the risky asset. The remaining fund is then invested in the risk-free asset. We do not restrict the sign of the investment decision $\{I_i(t)\}_{t \in \mathcal{T}}$. A negative investment decision $\{I_i(t)\}_{t \in \mathcal{T}}$ indicates a short sale of the risky asset.

Following the work in \cite{rogers2013optimal}, the fund process $\{X_i(t)\}_{t\in\mathcal{T}}$ of household $\mathsf{H}_i$ can be modeled using the following stochastic differential equation:
\begin{equation}
    \dif X_i(t)=[rX_i(t)+vI_i(t)-C_i(t)]\dif t+\sigma I_i(t)\dif B(t), \label{equ:fund-process}
\end{equation}
where $\{B(t)\}_{t\in\mathcal{T}}$ is a standard Brownian motion defined in a complete filtered probability space $(\Omega,\mathcal{F},\{\mathscr{F}(t)\}_{t\in\mathcal{T}},\mathbb{P})$, representing the randomness of the risky asset's price.

During the time horizon $\mathcal{T}$, both households aim to maximize the expected utilities of their terminal funds. Following the work in \cite{rogers2013optimal}, we define $\mathsf{H}_i$'s fund utility as
\begin{equation}
    \phi_i[X_i(T)]:=-\frac{1}{\alpha_i}\mathrm{e}^{-\alpha_iX_i(T)}, \label{equ:fund-utility}
\end{equation}
where $\alpha_i>0$ is called the \textit{risk aversion coefficient} of $\mathsf{H}_i$. The risk aversion coefficient $\alpha_i$ quantifies the household’s sensitivity to uncertainty in terminal wealth. As the risk aversion coefficient $\alpha_i$ increases, $\mathsf{H}_i$'s fund utility becomes more sensitive to changes in his/her terminal fund, i.e., a higher value of $\alpha_i$ corresponds to a stronger aversion to financial risk. The exponential utility follows the standard framework in intertemporal investment and consumption problems in \cite{rogers2013optimal}, which is commonly used due to its risk-averse properties and tractability in closed-form solutions. 

\begin{table}
\caption{Notations}
\label{tab:notation}
\centering
\setlength{\tabcolsep}{17pt}
\begin{tabular}{c|l}
\toprule
\textbf{Notation} & \multicolumn{1}{c}{\textbf{Meaning}} \\
\midrule
$r$ & Interest rate \\
$v$ & Excess return rate \\
$\sigma$ & Volatility \\
$\mathcal{T}$ & Time horizon \\
$\rho$ & Decay coefficient \\
$\theta$ & Herd coefficient \\
$\alpha_i$ & $\mathsf{H}_i$'s risk aversion coefficient \\
$\beta_i$ & $\mathsf{H}_i$'s diminishing marginal coefficient \\
$\gamma_i$ & $\mathsf{H}_i$'s consumption weight coefficient \\
$x_i$ & $\mathsf{H}_i$'s initial fund \\
$X_i$ & $\mathsf{H}_i$'s fund process \\
$I_i$ & $\mathsf{H}_i$'s investment decision \\
$C_i$ & $\mathsf{H}_i$'s consumption decision \\
$\phi_i$ & $\mathsf{H}_i$'s fund utility \\
$\psi_i$ & $\mathsf{H}_i$'s consumption utility \\
$D$ & Average deviation \\
\bottomrule
\end{tabular}
\end{table}

Different from the investment-only problem that focuses on their terminal funds, both households consider the cumulative utility over the time horizon $\mathcal{T}$ when making their consumption decisions. That is, they aim to maximize the integral of the consumption utility. Following the work in \cite{rogers2013optimal}, we define the time-dependent consumption utility of $\mathsf{H}_i$ as
\begin{equation}
    \psi_i[t,C_i(t)]:=-\frac{1}{\beta_i}\mathrm{e}^{-\rho rt-\beta_iC_i(t)}, \label{equ:consumption-utility}
\end{equation}
where $\beta_i>0$ is called the \textit{diminishing marginal coefficient} of $\mathsf{H}_i$. As the diminishing marginal coefficient $\beta_i$ increases, $\mathsf{H}_i$'s utility becomes more sensitive to changes in consumption levels, indicating a stronger sense of the burden of consumption and a tendency to avoid excessive consumption. 

In equation \eqref{equ:consumption-utility}, $\rho > 0$ is called the \textit{decay coefficient} of the consumption utility, with a higher decay coefficient $\rho$ indicating that the later consumption utility carries less weight. Specifically, when the decay coefficient $\rho = 1$, the exponential decay rate in \eqref{equ:consumption-utility} is the interest rate $r$. In this case, there is no inclination to consume either prematurely or belatedly \cite{wu2018equilibrium}. When the decay coefficient $\rho > 1$, the household exhibits a propensity to consume sooner. Conversely, when the decay coefficient $\rho < 1$, the household tends to defer consumption.

\subsection{The Leading Household's Optimization Problem}
Given the fund process in \eqref{equ:fund-process}, the fund utility in \eqref{equ:fund-utility}, and the consumption utility in \eqref{equ:consumption-utility}, the leading household $\mathsf{H}_2$'s objective functional can be expressed as
\begin{equation}
    \bar{J}(I_2,C_2)=\mathbb{E}\phi_2[X_2(T)]+\gamma_2\int_\mathcal{T}\psi_2[t,C_2(t)]\dif t,
    \label{equ:objective-functional-leading}
\end{equation}
where $\gamma_2>0$ is called the \textit{consumption weight coefficient}, and a higher consumption weight coefficient $\gamma_2$ indicates that $\mathsf{H}_2$ places more emphasis on consumption than investment, and vice versa.

In summary, the leading household $\mathsf{H}_2$'s optimal investment and consumption problem can be formulated as
\begin{equation}
    \sup_{{\{I_2(t)\}_{t\in\mathcal{T}},\{C_2(t)\}_{t \in \mathcal{T}}}}\bar{J}(I_2,C_2)\quad\text{s.t.\ }\eqref{equ:fund-process},\ 
    X_2(0)=x_2.
    \label{problem:leading-household}
\end{equation}

\subsection{The Following Household's Optimization Problem}
Next, we formulate the following household $\mathsf{H}_1$'s optimal investment and consumption problem. Following the work in \cite{wang2024ccc}, considering the investment herding, we use the average deviation to measure the distance between the two households' investment decisions, which is
\begin{equation}
    D(I_1,I_2)=\frac{1}{2}\int_\mathcal{T}\mathrm{e}^{\rho r(T-t)}[I_1(t)-I_2(t)]^2\dif t,
    \label{equ:average-deviation}
\end{equation}
where $\rho$ is the decay coefficient. In this work, to facilitate analysis, we assume that the household's decay coefficient $\rho$ is homogeneous in both \eqref{equ:consumption-utility} and \eqref{equ:average-deviation}, and we will study the heterogeneous scenario in our future work.

Like the leading household, the following household $\mathsf{H}_1$ aims to maximize the fund utility and the consumption utility. Furthermore, due to the investment herding, $\mathsf{H}_1$ also minimizes the average deviation. Therefore, the following household $\mathsf{H}_1$'s objective functional can be expressed as
\begin{equation}
    J(I_1,C_1)=\bar{J}(I_1,C_1)-\theta D(I_1,I_2),
\end{equation}
where $\bar{J}(I_1,C_1)$ is defined similarly to $\bar{J}(I_2,C_2)$ in \eqref{equ:objective-functional-leading}, and $\theta\geqslant0$ is called the \textit{herd coefficient} to address the tradeoff between the two different objectives, i.e., maximizing the fund and consumption utility $\bar{J}(I_1,C_1)$ and minimizing the average deviation $D(I_1,I_2)$. Specifically, when the herd coefficient $\theta=0$, $\mathsf{H}_1$'s optimal investment and consumption decisions are entirely independent of those of $\mathsf{H}_2$.

In summary, the following household $\mathsf{H}_1$'s optimal investment and consumption problem can be formulated as
\begin{equation}
    \sup_{{\{I_1(t)\}_{t\in\mathcal{T}},\{C_1(t)\}_{t \in \mathcal{T}}}}J(I_1,C_1)\quad\text{s.t.\ }\eqref{equ:fund-process},\ 
    X_1(0)=x_1.
    \label{problem:following-household}
\end{equation}

For tractability, we assume a frictionless market with continuous trading and no transaction costs. Additionally, external shocks such as macroeconomic downturns or sudden regulatory changes are not explicitly modeled, allowing us to focus on the endogenous effects of investment herding on household decisions. The extension to settings with market frictions and exogenous shocks is left for future research.

\section{The Solutions of the Optimal Investment and Consumption Decisions}
\label{sec:solution}
Using the variational method in optimal control, we can derive the optimal investment and consumption decisions for the leading and following households, respectively.

\subsection{The Leading Household's Optimal Decisions and the Following Household's Rational Decisions}
\begin{theorem}\label{the:leading-optimal-decision}
For the leading household $\mathsf{H}_2$, the optimal investment decision $\{\bar{I}_2(t)\}_{t\in\mathcal{T}}$ and the optimal consumption decision $\{\bar{C}_2(t)\}_{t\in\mathcal{T}}$ in \eqref{problem:leading-household} are given by
\begin{align}
    \bar{I}_2(t)&=\frac{v}{\alpha_2\sigma^2}\mathrm{e}^{r(t-T)},\quad \text{and}\\
    \bar{C}_2(t)&=\frac{1-\rho}{\beta_2}rt+\bar{k}_2,
    \label{equ:optimal-decision-leading}
\end{align}
respectively, where the parameter $\bar{k}_2$ in \eqref{equ:optimal-decision-leading} is given by
\begin{align}
    \bar{k}_2&=\left(\frac{\beta_2}{\alpha_2}+\frac{\mathrm{e}^{rT}-1}{r}\right)^{-1}\left[-rT+x_2\mathrm{e}^{rT}+\frac{\ln{\gamma_2}}{\alpha_2}\right.\notag\\
    &+\left.\frac{v^2T}{2\alpha_2\sigma^2}+\frac{(\rho-1)(\mathrm{e}^{rT}-rT-1)}{\beta_2r}\right].
\end{align}
\end{theorem}

\begin{proof}
Theorem \ref{the:leading-optimal-decision} is a special case of Theorem \ref{the:following-optimal-decision} when the herd coefficient $\theta=0$. See the proof of Theorem \ref{the:following-optimal-decision}.
\end{proof}

We call the optimal investment decision $\{\bar{I}_2(t)\}_{t\in\mathcal{T}}$ and consumption decision $\{\bar{C}_2(t)\}_{t\in\mathcal{T}}$ in \eqref{equ:optimal-decision-leading} the \textit{rational investment decision} and \textit{rational consumption decision} for $\mathsf{H}_2$ without others' influence, respectively.

Similarly to \eqref{equ:optimal-decision-leading}, we can define the rational investment and consumption decisions for $\mathsf{H}_1$, which are given by
\begin{align}
    \bar{I}_1(t)&=\frac{v}{\alpha_1\sigma^2}\mathrm{e}^{r(t-T)},\quad \text{and}\\
    \bar{C}_1(t)&=\frac{1-\rho}{\beta_1}rt+\bar{k}_1,
    \label{equ:rational-decision-following}
\end{align}
respectively, where the parameter $\bar{k}_1$ in \eqref{equ:rational-decision-following} is given by
\begin{align}
    \bar{k}_1&=\left(\frac{\beta_1}{\alpha_1}+\frac{\mathrm{e}^{rT}-1}{r}\right)^{-1}\left[-rT+x_1\mathrm{e}^{rT}+\frac{\ln{\gamma_1}}{\alpha_1}\right.\notag\\
    &+\left.\frac{v^2T}{2\alpha_1\sigma^2}+\frac{(\rho-1)(\mathrm{e}^{rT}-rT-1)}{\beta_1r}\right].
\end{align}

\subsection{The Following Household's Optimal Decisions}

\begin{theorem}\label{the:following-optimal-decision}
For the following household $\mathsf{H}_1$, the optimal investment decision $\{I_1^*(t)\}_{t\in\mathcal{T}}$ and the optimal consumption decision $\{C_1^*(t)\}_{t\in\mathcal{T}}$ in \eqref{problem:following-household} are given by
\begin{align}
    I_1^*(t)&=\frac{\eta\alpha_2\sigma^2\mathrm{e}^{(2-\rho) r(T-t)}+\theta}{\eta\alpha_1\sigma^2\mathrm{e}^{(2-\rho) r(T-t)}+\theta}\cdot\bar{I}_2(t),\quad\text{and}
    \label{equ:optimal-investment-decision-following}\\
    C_1^*(t)&=\frac{1-\rho}{\beta_1}rt+k^*_1,
    \label{equ:optimal-consumption-decision-following}
\end{align}
respectively, where the parameter $k^*_1$ is given by
\begin{align}
    k^*_1&=\left(\frac{\beta_1}{\alpha_1}+\frac{\mathrm{e}^{rT}-1}{r}\right)^{-1}\left\{-rT+x_1\mathrm{e}^{rT}+\frac{\ln{\gamma_1}}{\alpha_1}\right.\notag\\
    &+\left.\frac{v^2T}{2\alpha_1\sigma^2}+\frac{(\rho-1)(\mathrm{e}^{rT}-rT-1)}{\beta_1r}\right.\notag\\
    &-\left.\frac{\alpha_1\sigma^2}{2}\int_\mathcal{T}\mathrm{e}^{2r(T-t)}[I_1^*(t)-\bar{I}_1(t)]^2\dif t\right\},
    \label{equ:k_1^*}
\end{align}
and the parameter $\eta$ is given by
\begin{align}
    \eta&=\exp\left\{-\alpha_1x_1\mathrm{e}^{rT}-\alpha_1\int_\mathcal{T}\mathrm{e}^{r(T-t)}[vI_1^*(t)-C_1^*(t)]\dif t\right.\notag\\
    &+\left.\frac{\alpha_1^2\sigma^2}{2}\int_\mathcal{T}\mathrm{e}^{2r(T-t)}I_1^{*2}(t)\dif t\right\}.
    \label{equ:integral-constant}
\end{align}
\end{theorem}

\begin{proof}
Following the work in \cite{wang2024ccc}, we can derive the expression objective functional $J(I_1,C_1)$, and the variation $\delta J(I_1,\delta I_1,C_1,\delta C_1)$ and the second-order variation $\delta^2 J(I_1,\delta I_1,C_1,\delta C_1)$, which are given by \eqref{equ:objective-functional}, \eqref{equ:objective-functional-variation}, and \eqref{equ:objective-functional-variation-second}, respectively. According to the variational method \cite{kirk2004optimal}, let $\delta J(I_1,\delta I_1,C_1,\delta C_1)=0$, and we can obtain the necessary condition that the objective functional $J(I_1,C_1)$ reaches its supremum. Note that the second-order variation $\delta^2 J(I_1,\delta I_1,C_1,\delta C_1)$ in \eqref{equ:objective-functional-variation-second} is strictly negative. Therefore, $\delta J(I_1,\delta I_1,C_1,\delta C_1)=0$ is also a sufficient condition for the objective functional $J(I_1,C_1)$ to reach its supremum. Using the same method in \cite{wang2024ccc}, we can further obtain the optimal investment decision $\{I_1^*(t)\}_{t\in\mathcal{T}}$ in \eqref{equ:optimal-investment-decision-following}. 

\begin{figure*}[hb]
\vspace{-1mm}
\centering
\hrulefill
\small
\begin{align}
    J(I_1,C_1)&=-\frac{1}{\alpha_1}\exp\left\{-\alpha_1x_1\mathrm{e}^{rT}-\alpha_1\int_{\mathcal{T}}\mathrm{e}^{r(T-t)}[vI_1(t)-C_1(t)]\dif t+\frac{\alpha_1^2\sigma^2}{2}\int_{\mathcal{T}}\mathrm{e}^{2r(T-t)}I_1^2(t)\dif t\right\}\notag\\
    &-\frac{\gamma_1}{\beta_1}\int_\mathcal{T}\mathrm{e}^{-\rho rt-\beta_1C_1(t)}\dif t-\frac{\theta}{2}\int_{\mathcal{T}}\mathrm{e}^{\rho r(T-t)}[I_1(t)-\bar{I}_2(t)]^2\dif t.
    \label{equ:objective-functional}\\
    \delta J(I_1,\delta I_1,C_1,\delta C_1)&=\exp\left\{-\alpha_1x_1\mathrm{e}^{rT}-\alpha_1\int_{\mathcal{T}}\mathrm{e}^{r(T-t)}[vI_1(t)-C_1(t)]\dif t+\frac{\alpha_1^2\sigma^2}{2}\int_{\mathcal{T}}\mathrm{e}^{2r(T-t)}I_1^2(t)\dif t\right\}\notag\\
    &\cdot\left\{\int_\mathcal{T}\left[v\mathrm{e}^{r(T-t)}-\alpha_1\sigma^2\mathrm{e}^{2r(T-t)}I_1(t)\right]\delta I_1(t)\dif t-\int_\mathcal{T}\mathrm{e}^{r(T-t)}\delta C_1(t)\dif t\right\}\notag\\
    &+\gamma_1\int_\mathcal{T}\mathrm{e}^{-\rho rt-\beta_1C_1(t)}\delta C_1(t)\dif t-\theta\int_\mathcal{T}\mathrm{e}^{\rho r(T-t)}[I_1(t)-\bar{I}_2(t)]\delta I_1(t)\dif t.\label{equ:objective-functional-variation}\\
    \delta^2 J(I_1,\delta I_1,\delta I_1^2,C_1,\delta C_1,\delta C_1^2)&=-\alpha_1\exp\left\{-\alpha_1x_1\mathrm{e}^{rT}-\alpha_1\int_{\mathcal{T}}\mathrm{e}^{r(T-t)}[vI_1(t)-C_1(t)]\dif t+\frac{\alpha_1^2\sigma^2}{2}\int_{\mathcal{T}}\mathrm{e}^{2r(T-t)}I_1^2(t)\dif t\right\}\notag\\
    &\cdot\left\{\int_\mathcal{T}\left[v\mathrm{e}^{r(T-t)}-\alpha_1\sigma^2\mathrm{e}^{2r(T-t)}I_1(t)\right]\delta I_1(t)\dif t-\int_\mathcal{T}\mathrm{e}^{r(T-t)}\delta C_1(t)\dif t\right\}^2\notag\\
    &-\alpha_1\sigma^2\exp\left\{-\alpha_1x_1\mathrm{e}^{rT}-\alpha_1\int_{\mathcal{T}}\mathrm{e}^{r(T-t)}[vI_1(t)-C_1(t)]\dif t+\frac{\alpha_1^2\sigma^2}{2}\int_{\mathcal{T}}\mathrm{e}^{2r(T-t)}I_1^2(t)\dif t\right\}\notag\\
    &\cdot\int_\mathcal{T}\mathrm{e}^{2r(T-t)}[\delta I_1(t)]^2\dif t-\beta_1\gamma_1\int_\mathcal{T}\mathrm{e}^{-\rho rt-\beta_1C_1(t)}[\delta C_1(t)]^2\dif t-\theta\int_\mathcal{T}\mathrm{e}^{\rho r(T-t)}[\delta I_1(t)]^2\dif t.\label{equ:objective-functional-variation-second}
\end{align}
\end{figure*}

Next, we prove \eqref{equ:optimal-consumption-decision-following} and \eqref{equ:k_1^*}. From \eqref{equ:objective-functional-variation}, we have
\begin{equation}
    C_1^*(t)=\frac{1-\rho}{\beta_1}rt+\frac{1}{\beta_1}\left(\ln\frac{\gamma_1}{\eta}-rT\right).\label{equ:delta-C}
\end{equation}

Equation \eqref{equ:delta-C} indicates that the optimal consumption decision $\{C_1^*(t)\}$ is a linear function of time $t$. Denote the constant term as $k_1^*$, and we obtain \eqref{equ:optimal-consumption-decision-following}.

Then, substituting \eqref{equ:delta-C} into \eqref{equ:integral-constant}, we have
\begin{align}
    \eta&=\exp\left\{-\alpha_1x_1\mathrm{e}^{rT}-\alpha_1v\int_\mathcal{T}\mathrm{e}^{r(T-t)}I_1^*(t)\dif t\right.\notag\\
    &+\alpha_1\int_\mathcal{T}\mathrm{e}^{r(T-t)}\left[\frac{1-\rho}{\beta_1}rt+\frac{1}{\beta_1}\left(\ln\frac{\gamma_1}{\eta}-rT\right)\right]\dif t\notag\\
    &+\left.\frac{\alpha_1^2\sigma^2}{2}\int_\mathcal{T}\mathrm{e}^{2r(T-t)}I_1^{*2}(t)\dif t\right\}.
    \label{equ:integral-constant-C}
\end{align}

By solving \eqref{equ:integral-constant-C} for $\ln\eta$, we have
\begin{align}
    \ln\eta&=\left[1+\frac{\alpha_1}{\beta_1}\int_\mathcal{T}\mathrm{e}^{r(T-t)}\dif t\right]^{-1}\bigg\{\!-\!\alpha_1x_1\mathrm{e}^{rT}\notag\\
    &+\alpha_1\int_\mathcal{T}\mathrm{e}^{r(T-t)}\left[\frac{1-\rho}{\beta_1}rt+\frac{1}{\beta_1}\left(\ln\gamma_1-rT\right)\right]\dif t\notag\\
    &-\alpha_1v\int_\mathcal{T}\mathrm{e}^{r(T-t)}I_1^*(t)\dif t\notag\\
    &+\left.\frac{\alpha_1^2\sigma^2}{2}\int_\mathcal{T}\mathrm{e}^{2r(T-t)}I_1^{*2}(t)\dif t\right\}.\label{equ:ln-eta}
\end{align}

Substituting \eqref{equ:rational-decision-following} into \eqref{equ:ln-eta}, we have
\begin{align}
    \ln\eta&=\ln\gamma_1-rT-\beta_1\left(\frac{\beta_1}{\alpha_1}+\frac{\mathrm{e}^{rT}-1}{r}\right)^{-1}\label{equ:ln-eta-2}\notag\\
    &\cdot\left\{-rT+x_1\mathrm{e}^{rT}+\frac{\ln{\gamma_1}}{\alpha_1}\right.\notag\\
    &+\left.\frac{v^2T}{2\alpha_1\sigma^2}+\frac{(\rho-1)(\mathrm{e}^{rT}-rT-1)}{\beta_1r}\right.\notag\\
    &-\left.\frac{\alpha_1\sigma^2}{2}\int_\mathcal{T}\mathrm{e}^{2r(T-t)}[I_1^*(t)-\bar{I}_1(t)]^2\dif t\right\}.
\end{align}

Finally, substituting \eqref{equ:ln-eta-2} into \eqref{equ:delta-C}, we obtain \eqref{equ:k_1^*}.
\end{proof}

\subsection{Analysis of the Rational and Optimal Decisions}
Based on Theorem \ref{the:leading-optimal-decision} and Theorem \ref{the:following-optimal-decision}, we theoretically analyze the rational investment and consumption decisions and the optimal investment and consumption decisions considering investment herding, and compare them with those decisions of the investment-only problem.

\subsubsection{The Investment Decisions}
\noindent\textbf{Rational Investment Decision:}\quad Theorem \ref{the:leading-optimal-decision} shows that the rational investment decision $\{\bar{I}_i(t)\}_{t\in\mathcal{T}}$ is the same as that of investment-only problem. This suggests that rational investment decision $\{\bar{I}_i(t)\}_{t\in\mathcal{T}}$ is independent of the consumption decision $\{\bar{C}_i(t)\}_{t\in\mathcal{T}}$. That is, the consumption decision does not influence the investment decision without considering investment herding. This conclusion is called the Fisher Separation Theorem in macroeconomics \cite{fisher1930theory}. An intuitive explanation of the Fisher Separation Theorem is that, since financial assets offer all possible combinations of risk and return, households can make investment decisions independently of their current and future funds. In other words, they do not need to know their consumption requirements when determining their optimal investment strategy.

The rational investment decision $\{\bar{I}_i(t)\}_{t\in\mathcal{T}}$ is an exponential function of time $t$ with an exponential growth rate equal to the interest rate $r$, and is proportional to the excess return rate $v$, inversely proportional to the volatility $\sigma$, and the household's risk aversion coefficient $\alpha_i$, which is consistent with the meanings of the parameters. 

\noindent\textbf{Optimal Investment Decision Considering Investment Herding:}\quad Theorem \ref{the:following-optimal-decision} shows that the expression of the following household $\mathsf{H}_1$'s optimal investment decision $\{I_1^*(t)\}_{t\in\mathcal{T}}$ in \eqref{equ:optimal-investment-decision-following} is the same as that in the investment-only problem, while the parameter $\eta$ in \eqref{equ:integral-constant} is different. In \eqref{equ:integral-constant}, the optimal consumption decision $\{C_1^*(t)\}_{t\in\mathcal{T}}$ affects the value of the parameter $\eta$, which further influences the optimal investment decision $\{I_1^*(t)\}_{t\in\mathcal{T}}$. This suggests that, when investment herding exists, the optimal investment and consumption decisions are no longer independent. Due to investment herding, households cannot achieve any combination of risk and return through financial assets, as their investment decisions are always distorted by the decisions of others. Therefore, their investment decisions are influenced by their current funds and consumption decisions.

When the optimal consumption decision $\{C_1^*(t)\}_{t\in\mathcal{T}}$ is strictly positive, from \eqref{equ:integral-constant}, the parameter $\eta$ is greater than that in the investment-only problem. Thus, from \eqref{equ:optimal-consumption-decision-following}, the optimal investment decision $\{I_1^*(t)\}_{t\in\mathcal{T}}$ aligns more with the rational investment decision $\{\bar{I}_1(t)\}_{t\in\mathcal{T}}$, indicating that the intensity of investment herding reduces when considering household consumption.

\subsubsection{The Consumption Decisions}
\noindent\textbf{Rational Consumption Decision:}\quad Theorem \ref{the:leading-optimal-decision} shows that the rational consumption decision $\{\bar{C}_i(t)\}_{t\in\mathcal{T}}$ is a linear function of time $t$, where the linear coefficient is proportional to the interest rate $r$ and inversely proportional to the household's diminishing marginal coefficient $\beta_i$. This means that the higher the interest rate $r$ and the smaller the diminishing marginal coefficient $\beta_i$, the faster $\{\bar{C}_i(t)\}_{t\in\mathcal{T}}$ changes over time. On one hand, as the interest rate $r$ increases, the household's fund process $\{X_i(t)\}_{t\in\mathcal{T}}$ changes faster over time, leading to a faster changing rate in consumption over time. On the other hand, as the diminishing marginal coefficient $\beta_i$ increases, the household tends to avoid excessive consumption, leading to a slower changing rate in consumption.

Furthermore, Theorem \ref{the:leading-optimal-decision} shows that when the decay coefficient $\rho>1$, which means that $\mathsf{H}_i$ exhibits a propensity to consume sooner, the consumption decision at earlier times is higher than at later times, and thus the rational consumption decision $\{\bar{C}_i(t)\}_{t\in\mathcal{T}}$ decreases with time $t$. Conversely, when the decay coefficient $\rho<1$, which means that $\mathsf{H}_i$ tends to defer consumption, the consumption decision at earlier times is lower than at later times, and thus the rational consumption decision $\{\bar{C}_i(t)\}_{t\in\mathcal{T}}$ increases with time $t$. Specifically, when the decay coefficient $\rho=1$, the rational consumption decision of $\mathsf{H}_i$ is time-invariant.

Additionally, Theorem \ref{the:leading-optimal-decision} shows that the rational consumption decision $\{\bar{C}_i(t)\}_{t\in\mathcal{T}}$ increases with the household's initial fund $x_i$ and consumption weight coefficient $\gamma_i$, which is consistent with the meanings of these parameters. Additionally, as the excess return $v$ increases and the volatility $\sigma$ decreases, the household has more funds, leading to a higher rational consumption decision $\{\bar{C}_i(t)\}_{t\in\mathcal{T}}$.

\noindent\textbf{Optimal Consumption Decision Considering Investment Herding:}\quad Theorem \ref{the:following-optimal-decision} shows that the following household $\mathsf{H}_1$'s optimal consumption decision $\{C_1^*(t)\}_{t\in\mathcal{T}}$ is also a linear function of time $t$, whose linear coefficient is the same as that of the rational consumption decision $\{\bar{C}_1(t)\}_{t\in\mathcal{T}}$, while the constant term $k_1^*$ differs from $\bar{k}_1$. We will further analyze the difference between the optimal consumption decision $\{C_1^*(t)\}_{t\in\mathcal{T}}$ and the rational consumption decision $\{\bar{C}_1(t)\}_{t\in\mathcal{T}}$ in Section \ref{sec:theoretical-analysis}.

\section{The Crowding-Out Effect of Investment Herding on Consumption}
\label{sec:theoretical-analysis}
In this section, we theoretically analyze how investment herding crowds out household consumption. 

\subsection{The Crowding-Out Consumption}
Comparing \eqref{equ:optimal-consumption-decision-following} with \eqref{equ:rational-decision-following}, we can find that the optimal consumption decision $\{C_1^*(t)\}_{t\in\mathcal{T}}$ and the rational consumption decision $\{\bar{C}_1(t)\}_{t\in\mathcal{T}}$ only differ by a constant, which is denoted as 
\begin{equation}
    \mathscr{C}:=\bar{C}_1(t)-C_1^*(t).
    \label{equ:crowding-out-consumption-definition}
\end{equation}

In the following, we will prove that the constant $\mathscr{C}$ is larger than or equal to zero. That is, the optimal consumption decision $\{C_1^*(t)\}_{t\in\mathcal{T}}$ is no higher than the rational consumption decision $\{\bar{C}_1(t)\}_{t\in\mathcal{T}}$ due to investment herding. Therefore, we call the constant $\mathscr{C}$ the \textit{crowding-out consumption} for the following household $\mathsf{H}_1$. First, we derive the analytical expression of the crowding-out consumption $\mathscr{C}$.

\begin{theorem}\label{the:crowding-out-consumption}
The crowding-out consumption $\mathscr{C}$ is given by
\begin{align}
    \mathscr{C}&=\frac{\alpha_1\sigma^2}{2}\left(\frac{\beta_1}{\alpha_1}+\frac{\mathrm{e}^{rT}-1}{r}\right)^{-1}\notag\\
    &\cdot\int_\mathcal{T}\mathrm{e}^{2r(T-t)}[I_1^*(t)-\bar{I}_1(t)]^2\dif t\geqslant0.
    \label{equ:crowding-out-consumption}
\end{align}
\end{theorem}

\begin{proof}
Combining \eqref{equ:rational-decision-following} and \eqref{equ:optimal-consumption-decision-following}, we can obtain \eqref{equ:crowding-out-consumption}.
\end{proof}

Following the work in \cite{wang2024ccc}, when the herd coefficient $\theta > 0$ and the two households' risk aversion coefficients $\alpha_1$ and $\alpha_2$ are not equal, $\mathsf{H}_1$'s optimal investment decision $\{I_1^*(t)\}_{t\in\mathcal{T}}$ differs from the rational investment decision $\{\bar{I}(t)\}_{t\in\mathcal{T}}$. Therefore, the definite integral in \eqref{equ:crowding-out-consumption} is strictly positive, and the crowding-out consumption $\mathscr{C}$ is also strictly positive. So far, we have theoretically proved the existence of the crowding-out effect of investment herding on household consumption.

\subsection{The Influence of Parameters on the Crowding-Out Consumption}
Next, we study the influence of parameters, including the interest rate $r$, the excess return rate $v$, and the volatility $\sigma$, on the crowding-out consumption $\mathscr{C}$. We first consider the special case where the herd coefficient $\theta$ approaches infinity in Section \ref{sec:sub-theoretical}, and then study the general cases using numerical experiments in Section \ref{sec:numerical-experiments}.

\subsubsection{Analysis of the Special Case}
\label{sec:sub-theoretical}
Following the work in \cite{wang2024ccc}, when the herd coefficient $\theta$ approaches infinity, the following household $\mathsf{H}_1$'s optimal investment decision $\{I_1^*(t)\}_{t\in\mathcal{T}}$ is equal to the leading household $\mathsf{H}_2$'s rational investment decision $\{\bar{I}_2(t)\}_{t\in\mathcal{T}}$. In this case, from \eqref{equ:optimal-decision-leading} and \eqref{equ:crowding-out-consumption}, the crowding-out consumption, denoted as $\tilde{\mathscr{C}}$, is given by
\begin{equation}
    \tilde{\mathscr{C}}=\left(\frac{\alpha_1}{\alpha_2}-1\right)^2\left(\frac{\beta_1}{\alpha_1}+\frac{\mathrm{e}^{rT}-1}{r}\right)^{-1}\frac{v^2T}{2\alpha_1\sigma^2}.\label{equ:crowding-out-consumption-infinity}
\end{equation}

\noindent\textbf{Interest Rate:}\quad From \eqref{equ:crowding-out-consumption-infinity}, we can prove that $\frac{\partial \tilde{\mathscr{C}}}{\partial r}<0$, i.e., the crowding-out consumption $\tilde{\mathscr{C}}$ decreases with the interest rate $r$. This can be explained by the following macroeconomic principle. When the interest rate $r$ rises, households are more inclined to invest their funds in risk-free assets to earn higher interest income and reduce the funds for consumption \cite{kozlov2023effect}. Therefore, the constant terms of the rational consumption decision $\bar{k}_1$ in \eqref{equ:rational-decision-following} and the optimal consumption decision $k^*_1$ in \eqref{equ:optimal-consumption-decision-following} both decrease, thereby reducing the crowding-out consumption $\tilde{\mathscr{C}}$.

\noindent\textbf{Excess Return Rate:}\quad From \eqref{equ:crowding-out-consumption-infinity}, we can prove that $\frac{\partial \tilde{\mathscr{C}}}{\partial v}>0$, i.e., the crowding-out consumption $\tilde{\mathscr{C}}$ increases with the excess return rate $v$. This can be explained by the following macroeconomic principle. An increase in the excess return rate $v$ leads to higher earnings for households from the financial market, which in turn raises their propensity to consume \cite{poterba2000stock}. Therefore, the constant terms of the rational consumption decision $\bar{k}_1$ in \eqref{equ:rational-decision-following} and the optimal consumption decision $k^*_1$ in \eqref{equ:optimal-consumption-decision-following} both increase, thereby enlarging the crowding-out consumption $\tilde{\mathscr{C}}$.

\noindent\textbf{Volatility:}\quad From \eqref{equ:crowding-out-consumption-infinity}, we can prove that $\frac{\partial \tilde{\mathscr{C}}}{\partial \sigma}<0$, i.e., the crowding-out consumption $\tilde{\mathscr{C}}$ decreases with the volatility $\sigma$. This can be explained by the following macroeconomic principle. A high volatility $\sigma$ leads to increased uncertainty about future income, thereby households tend to reduce current consumption expenditure to prepare for unforeseen needs. Conversely, a low volatility $\sigma$ implies a more stable market, and households may be more willing to consume because they have greater confidence in their future income \cite{choudhry2003stock}. Therefore, higher volatility $\sigma$ causes the constant terms of the rational consumption decision $\bar{k}_1$ in \eqref{equ:rational-decision-following} and the optimal consumption decision $k^*_1$ in \eqref{equ:optimal-consumption-decision-following} to decrease, thereby reducing the crowding-out consumption $\tilde{\mathscr{C}}$.

\subsubsection{Analysis of the General Cases}
\label{sec:numerical-experiments}
Next, we conduct numerical experiments to study the influence of parameters on the crowding-out consumption under the general cases with arbitrary herd coefficient $\theta$.

Following the work in \cite{yuen2001estimation}, we assign the risk aversion coefficients $\alpha_1 = 0.2$ and $\alpha_2 = 0.4$, and the diminishing marginal coefficients $\beta_1 = 0.2$ and $\beta_2 = 0.4$ to the two households. For both households, we set the consumption weight coefficients as $\gamma_1=\gamma_2=1$ and the initial funds as $x_1 = x_2 = 1$. Following the work in \cite{wang2024ccc}, we set the herd coefficient as $\theta=0.01$ and the decay coefficient as $\rho=1$, respectively. We vary the values of the interest rate $r$, the excess return rate $v$, and the volatility $\sigma$ across $[0.005, 0.025]$, $[0.05, 0.25]$, and $[0.05, 0.25]$, respectively. We observe the same trend for other values. 

\begin{figure}[!t]
\centering
\subfloat[Interest rate]{\includegraphics[width=0.33\linewidth]{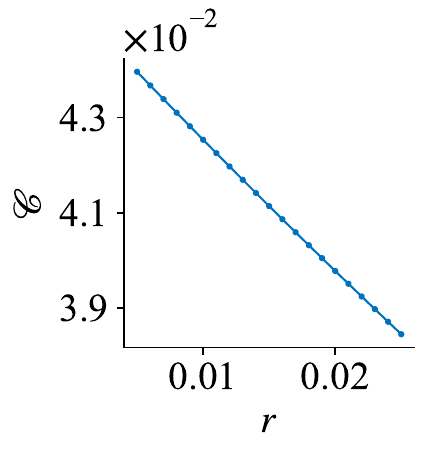} \label{fig:fig2a}}
\subfloat[Excess return rate]{\includegraphics[width=0.33\linewidth]{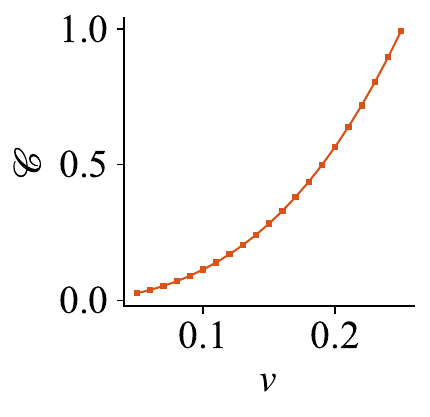} \label{fig:fig2b}}
\subfloat[Volatility]{\includegraphics[width=0.33\linewidth]{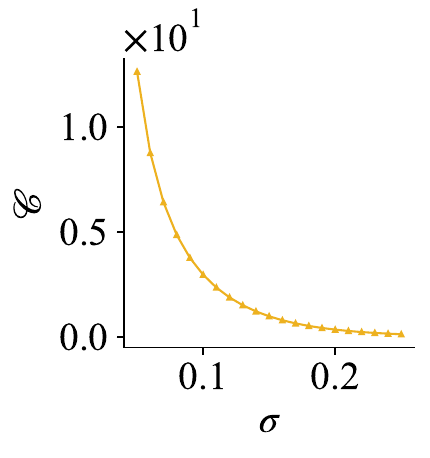} \label{fig:fig2c}}
\caption{The influences of parameters on the crowding-out consumption $\mathscr{C}$.}
\label{fig:fig2}
\end{figure}

We calculate the crowding-out consumption $\mathscr{C}$ using \eqref{equ:crowding-out-consumption}. The experimental results are in Fig. \ref{fig:fig2}. From Fig. \ref{fig:fig2}, we can find that the crowding-out consumption $\mathscr{C}$ decreases with the interest rate $r$ and the volatility $\sigma$, and increases with the excess return rate $v$. The above results validate our theoretical analysis in Section \ref{sec:theoretical-analysis}.

\section{Real Data Test}
\label{sec:real-data-validation}
To verify that our analysis of the crowding effect in Section \ref{sec:theoretical-analysis} aligns with reality, we conduct an empirical analysis using real data from China.

\subsection{Data Collection and Preprocessing}
In behavioral economics, the Cross-Sectional Standard Deviation (CSSD) is a commonly used metric for measuring investment herding \cite{christie1995following}. A higher CSSD indicates a weaker degree of investment herding, which corresponds to a smaller herd coefficient $\theta$. We obtain quarterly CSSD data from the CSMAR database \cite{csmar}, covering the period from October 2014 to September 2024. Our database selection is motivated by its broad coverage and reliability in capturing investor sentiment in China's stock market. 

To quantify household consumption, we use the cumulative per capita Consumption Expenditure (CE) of urban residents in China. We obtain quarterly CE data for the same period as CSSD from the National Bureau of Statistics \cite{nbs}. Due to the large magnitude of the CE values, we calculate the growth rate relative to the previous quarter. 

Additionally, we collect data on the risk-free rate $r$, excess return rate $v$, and volatility $\sigma$ of the Shanghai Stock Exchange Composite Index from \cite{csmar} and compute their quarterly values. To account for external factors influencing household consumption, we also obtain quarterly data on the Total Retail Sales of consumer goods (TRS) from \cite{nbs}, which serves as an indicator of overall consumer demand in the economy. Finally, we perform linear normalization on each of the above data to the $[0,1]$ interval. 

\subsection{Regression Results and Analysis}
To examine the dynamic impact of investment herding on household decisions, we construct a multiple linear regression model that qualitatively captures the crowding-out effect of investment herding on household consumption and the influence of parameters on the crowding-out effect, while controlling for other economic variables. This specification allows us to isolate the effect of the CSSD metric on household consumption and determine whether the crowding-out effect persists over time. The multiple linear regression model is given by
\begin{equation}
    \text{CE}=a_0+a_1\text{TRS}+a_2r+a_3v+a_4\sigma+a_5\text{CSSD}+\varepsilon,\label{equ:regression}
\end{equation}
where $\varepsilon$ is the noise. The regression results are in Table \ref{tab:regression-results}.

From Table \ref{tab:regression-results}, first, the metric of investment herding, CSSD, is positively correlated with the household consumption CE, which indicates that a smaller CSSD, i.e., stronger investment herding, leads to lower household consumption CE, thereby verifying the existence of the crowding-out effect of investment herding on household consumption. Also, the interest rate $r$ and volatility $\sigma$ are negatively correlated with the household consumption CE, which further negatively impacts the crowding-out consumption $\mathscr{C}$. Finally, the excess return rate $v$ is positively correlated with the household consumption CE, thereby exerting a positive impact on the crowding-out consumption $\mathscr{C}$. Furthermore, from the t-statistics, the CSSD and TRS metrics are more statistically significant than the interest rate $r$, the excess return rate $v$, and the volatility $\sigma$, which indicates that investment herding exhibits a stronger impact on the crowding-out effect on household consumption compared to market parameters. This is consistent with prior empirical findings in behavioral finance \cite{bikhchandani2000herd, chang2014herd}, which highlight the dominant role of herding effects in shaping investor behavior. To assess the robustness of our results, we conduct an additional test to replace CSSD with an alternative investment herding metric, the Cross-Sectional Absolute Deviation (CSAD) \cite{chang2000examination}, and the multiple linear regression model transforms into
\begin{equation}
    \text{CE}=b_0+b_1\text{TRS}+b_2r+b_3v+b_4\sigma+b_5\text{CSAD}+\varepsilon.\label{equ:regression_robust}
\end{equation}
The regression results are in Table \ref{tab:regression-results-robust}, which confirms that our findings remain qualitatively consistent. The above results validate our theoretical analysis in Section \ref{sec:theoretical-analysis}.

\begin{table*}[!t]
\centering
\begin{minipage}{0.48\textwidth}
\centering
\setlength{\tabcolsep}{3.5pt}
\begin{threeparttable}
\caption{Regression Results of Model \eqref{equ:regression}}
\label{tab:regression-results}
\begin{tabular}{lrrr}
\toprule
\textbf{Variable} & \textbf{Coefficient} & \textbf{Standard Error} & \textbf{t-Statistic} \\ 
\midrule
Intercept & $-0.348$ & $0.210$ & $-1.661$ \\
TRS & $1.248$ & $0.210$ & $5.941$ \\
Interest rate $r$ & $-0.125$ & $0.268$ & $-0.467$ \\
Excess return rate $v$ & $0.091$ & $0.270$ & $0.336$ \\
Volatility $\sigma$ & $-0.121$ & $0.271$ & $-0.446$ \\
\underline{CSSD} & $0.253$ & $0.180$ & $1.403$ \\
\midrule
\multicolumn{4}{l}{\textbf{Model Statistics}} \\
R-squared & $0.522$ & \multicolumn{1}{l}{F-Statistic} & $7.417$ \\
Adjusted R-squared & $0.451$ & \multicolumn{1}{l}{Observations} & $40$ \\
\bottomrule
\end{tabular}
\end{threeparttable}
\end{minipage}
\hfill
\begin{minipage}{0.48\textwidth}
\centering
\setlength{\tabcolsep}{3.5pt}
\begin{threeparttable}
\caption{Regression Results of Model \eqref{equ:regression_robust}}
\label{tab:regression-results-robust}
\begin{tabular}{lrrr}
\toprule
\textbf{Variable} & \textbf{Coefficient} & \textbf{Standard Error} & \textbf{t-Statistic} \\ 
\midrule
Intercept & $-0.376$ & $0.208$ & $-1.807$ \\
TRS & $1.313$ & $0.212$ & $6.181$ \\
Interest rate $r$ & $-0.165$ & $0.264$ & $-0.625$ \\
Excess return rate $v$ & $0.052$ & $0.266$ & $0.196$ \\
Volatility $\sigma$ & $-0.286$ & $0.293$ & $-0.977$ \\
\underline{CSAD} & $0.426$ & $0.242$ & $1.759$ \\
\midrule
\multicolumn{4}{l}{\textbf{Model Statistics}} \\
R-squared & $0.536$ & \multicolumn{1}{l}{F-Statistic} & $7.863$ \\
Adjusted R-squared & $0.468$ & \multicolumn{1}{l}{Observations} & $40$ \\
\bottomrule
\end{tabular}
\end{threeparttable}
\end{minipage}
\vspace{-4mm}
\end{table*}

\section{Conclusion}
\label{sec:conclusion}
In this work, we quantitatively study the impact of investment herding on household consumption in the field of macroeconomics. We formulate a leader-follower optimal investment and consumption problem, and based on the optimal control theory, we derive the analytical solutions for the optimal investment and consumption decisions of both leading and following households. We theoretically demonstrate the existence of the crowding-out effect of investment herding on consumption and analyze how parameters influence the crowding-out effect. We show that a higher interest rate and volatility tend to reduce the crowding-out effect, whereas a higher excess return rate amplifies the crowding-out effect. We validate our analysis using a real data test.

\bibliography{bibliology}
\bibliographystyle{IEEEtran}

\end{document}